\documentclass[12pt]{article}
\usepackage{graphicx}          
\usepackage{amsfonts,amsmath,amssymb,mathrsfs,appendix}

\newtheorem{theorem}{Theorem}

\newtheorem{definition}[theorem]{Definition}
\newtheorem{example}[theorem]{Example}

\newtheorem{lemma}[theorem]{Lemma}

\newtheorem{proposition}[theorem]{Proposition}
\newtheorem{remark}[theorem]{Remark}

\newenvironment{proof}[1][Proof]{\textbf{#1.} }{\ \rule{0.5em}{0.5em}}

\title{The  Transfer Function of Generic Linear Quantum Stochastic Systems Has a Pure Cascade Realization\thanks{This research was supported by the Australian Research Council}}
\author{Hendra I. Nurdin,~Symeon Grivopoulos,~and~Ian~R.~Petersen\footnote{H. I. Nurdin is with the School of Electrical Engineering and Telecommunications, UNSW Australia, Sydney NSW 2052, Australia. Email: h.nurdin@unsw.edu.au. S. Grivopoulos and I. R. Petersen are with the School of Engineering and Information Technology, UNSW Canberra, Canberra BC 2610, Australia.}}

\begin{document}

\maketitle

\begin{abstract}
This paper establishes that generic linear quantum stochastic systems have a pure cascade realization of their transfer function, generalizing an earlier result established only for the special class of completely passive linear quantum stochastic systems. In particular, a cascade realization therefore  exists for generic active linear quantum stochastic systems  that require an external source of quanta to operate.
The results  facilitate a simplified realization of generic linear quantum stochastic systems for applications such as coherent feedback control and optical filtering. The key tools that are developed are algorithms for symplectic QR and Schur decompositions. It is shown that generic  real square matrices of even dimension can be transformed into a lower $2 \times 2$ block triangular form by a symplectic similarity transformation.  The linear algebraic results herein may be of independent interest for applications beyond the problem of transfer function realization for quantum systems. Numerical examples are included to illustrate the main results. In particular, one example  describes an equivalent realization of the transfer function of a nondegenerate parametric amplifier as the cascade interconnection of two degenerate parametric amplifiers with an additional outcoupling mirror.
\end{abstract}


\section{Introduction}
\label{sec:intro}

The class of linear quantum stochastic systems \cite{BE08,JNP06,NJD08,GJN10} represents multiple distinct open quantum harmonic oscillators that are coupled linearly to one another and also to external Gaussian fields, e.g., coherent laser beams, and whose dynamics can be conveniently and completely summarized in the Heisenberg picture of quantum mechanics in terms of a quartet of matrices $A,B,C,D$, analogous to those used in modern control theory for linear systems. As such, they can  be viewed as a quantum analogue of classical linear stochastic systems and are encountered in practice, for instance, as models for optical parametric amplifiers \cite[Chapters 7 and 10]{GZ04}. However, due to the constraints imposed by quantum mechanics, the matrices $A,B,C,D$ in a linear quantum stochastic system cannot be arbitrary, a restriction not encountered in the classical setting. In fact, as derived in \cite{JNP06}  for a certain fixed choice of $D$, it is required that $A$ and $B$ satisfy a certain non-linear equality constraint, and $B$ and $C$ satisfy a linear equality constraint. These constraints on $A,B,C,D$ are referred to as {\em physical realizability} constraints \cite{JNP06}.

A number of applications of linear quantum stochastic systems have been theoretically proposed or experimentally demonstrated in the literature. In particular, they can serve as coherent feedback controllers \cite{JNP06,NJP07b}, i.e., feedback controllers that are themselves quantum systems. In this context,  they have been shown to be theoretically effective for cooling of an optomechanical resonator \cite{HM12}, can modify the characteristics of squeezed light produced experimentally by an optical parametric oscillator  (OPO) \cite{CTSAM13}, and, in the setting of microwave superconducting circuits, a linear quantum stochastic system in the form of a Josephson parametric amplifier (JPA) operated in the linear regime has been experimentally demonstrated to be able to rapidly reshape the dynamics of a superconducting electromechanical circuit (EMC) \cite{KAKKCSL13}. Linear quantum stochastic systems can also be used as optical filters for various input signals, including non-Gaussian input signals like single photon and multi-photon states. As filters they can be used to modify the wavepacket shape of single and multi-photon sources \cite{ZJ13,Zhang14}. Also, linear quantum stochastic  systems can dissipatively generate Gaussian cluster states \cite{KY12} as an important component of continuous-variable one way quantum computers \cite{MLGWRN06}.

In certain quantum control problems, such as in  coherent feedback $H^{\infty}$ \cite{JNP06} and LQG \cite{NJP07b} control  problems, the latter being adapted for addressing an optomechanical cooling problem in \cite{HM12}, the important feature of the controller is its transfer function $T(s)=C(sI-A)^{-1}B+D$ rather than the system matrices $(A,B,C,D)$. Therefore, an important issue in the implementation of coherent feedback controllers is how to realize a controller with a certain transfer function from a bin of basic linear quantum (optical) devices. This is a special case of the problem of network synthesis of linear quantum stochastic systems addressed in \cite{NJD08,Nurd10a,Nurd10b}. In particular, it was shown in \cite{Nurd10b}, generalizing the results of \cite{Pet09,Pet12}, that the transfer function of all linear quantum stochastic systems which are completely passive can be realized by a cascade of one degree of freedom linear quantum stochastic systems. Completely passive here means that the system can be realized using only passive linear optical devices which do not need an external source of quanta for their operation. The question of whether cascade realizations exist for general linear quantum stochastic systems has remained an open problem. The contribution of this paper is to resolve this question by proving that, generically, linear quantum stochastic systems do possess a pure cascade realization. This is significant from a practical point of view, as it allows for a simpler realization of generic linear quantum stochastic systems.

The remainder of the paper is organized as follows. Section \ref{sec:preliminaries} introduces the notation and gives an overview of linear quantum stochastic systems and the associated realization theory. Section \ref{sec:sym-GS} presents a symplectic QR decomposition algorithm. The results of Section \ref{sec:sym-GS} form the basis for a symplectic Schur decomposition algorithm that is presented in Section \ref{sec:pure-cascading} and used to show that the transfer function of generic linear quantum stochastic systems can be realized by pure cascading. Finally, Section \ref{sec:conclu} summarizes the contributions of the paper.

\section{Preliminaries}
\label{sec:preliminaries}

\subsection{Notation}
We will use the following notation: $\imath=\sqrt{-1}$, $^*$ denotes the adjoint of a linear operator
as well as the conjugate of a complex number. If $A=[a_{jk}]$ then $A^{\#}=[a_{jk}^*]$, and $A^{\dag}=(A^{\#})^{\top}$, where $(\cdot)^{\top}$ denotes matrix transposition.  $\Re\{A\}=(A+A^{\#})/2$ and $\Im\{A\}=\frac{1}{2\imath}(A-A^{\#})$.
We denote the identity matrix by $I$ whenever its size can be
inferred from context and use $I_{n}$ to denote an $n \times n$
identity matrix. Similarly, $0_{m \times n}$ denotes  a $m \times n$ matrix with zero
entries but drop the subscript when its dimension  can be determined from context. We use
${\rm diag}(M_1,M_2,\ldots,M_n)$  to denote a block diagonal matrix with
square matrices $M_1,M_2,\ldots,M_n$ on its diagonal, and ${\rm diag}_{n}(M)$
denotes a block diagonal matrix with the
square matrix $M$ appearing on its diagonal blocks $n$ times. Also, we
will let $\mathbb{J}=\left[\begin{array}{rr}0 & 1\\-1&0\end{array}\right]$ and $\mathbb{J}_n=I_{n} \otimes \mathbb{J}={\rm diag}_n(\mathbb{J})$.

\subsection{The class of linear quantum stochastic systems}
\label{sec:linear-summary}
Let $x=(q_1,p_1,q_2,p_2,\ldots,q_n,p_n)^T$ denote a vector of  the canonical position and
momentum operators of a {\em many degrees of freedom quantum
harmonic oscillator} satisfying the canonical commutation
relations (CCR)  $[q_i,p_j]=2 \imath \delta_{ij}$, $[q_i,q_j]=0$, and $[p_i,p_j]=0$ for $i,j=1,2,\ldots,n$, or, more compactly, $xx^T-(xx^T)^T=2\imath \mathbb{J}_n$. A {\em linear quantum stochastic
system} \cite{JNP06,NJP07b,NJD08} $G$ is a quantum system defined by three {\em parameters}:
(i) A quadratic Hamiltonian $H=\frac{1}{2} x^T R x$ with $R=R^T \in
\mathbb{R}^{2n \times 2n}$, (ii) a coupling operator $L=Kx$, where $K$ is
an $m \times 2n$ complex matrix, and (iii) a unitary $m \times m$
scattering matrix $S$. For shorthand, we write $G=(S,L,H)$ or $G=(S,Kx,\frac{1}{2} x^TRx)$. The
time evolution $x(t)$ of $x$ in the Heisenberg picture ($ t
\geq 0$) is given by the quantum stochastic differential equation
(QSDE) (see \cite{BE08,JNP06,NJD08}):
\begin{eqnarray*}
dx(t) &=&  A_0x(t)dt+ B_0\left[ \begin{array}{c} d\mathcal{A}(t)
\\ d\mathcal{A}(t)^{\#} \end{array}\right];
x(0)=x, \notag\\
dY(t) &=&  C_0 x(t)dt+  D_0 d\mathcal{A}(t), \label{eq:qsde-out}
\end{eqnarray*}
with $A_0=2\mathbb{J}_n(R+\Im\{K^{\dag}K\})$, $B_0=2\imath \mathbb{J}_n [\begin{array}{cc}
-K^{\dag}S & K^TS^{\#}\end{array}]$,
$C_0=K$, and $D_0=S$. Here $Y(t)=(Y_1(t),\ldots,Y_m(t))^{\top}$ is a vector of
continuous-mode bosonic {\em output fields} that results from the interaction of the quantum
harmonic oscillators and the incoming continuous-mode bosonic quantum fields in
the $m$-dimensional vector $\mathcal{A}(t)$. Note that the dynamics of $x(t)$ is linear, and
$Y(t)$ depends linearly on $x(s)$, $0 \leq s \leq t$. We refer to $n$ as the {\em
degrees of freedom} of the system or, more simply, the {\em degree} of the system.

Following \cite{JNP06}, it will be  convenient to write the dynamics in quadrature form as
\begin{align}
dx(t)&=Ax(t)dt+Bdw(t);\, x(0)=x. \nonumber\\
dy(t)&= C x(t)dt+ D dw(t), \label{eq:qsde-out-quad}
\end{align}
with $w(t) =2(\Re\{\mathcal{A}_1(t)\},\Im\{\mathcal{A}_1(t)\},\ldots,\Re\{\mathcal{A}_m(t)\},$ $\Im\{\mathcal{A}_m(t)\})^{\top}$ and
$y(t) = 2(\Re\{Y_1(t)\},$ $\Im\{Y_1(t)\},\ldots,$ $\Re\{Y_m(t)\},\Im\{Y_m(t)\})^{\top}$.
Here, the real matrices $A$, $B$, $C$, $D$ are in  one-to-one correspondence
with $A_0,B_0,C_0,D_0$. Also, $w(t)$ is taken to be in a vacuum state where it
satisfies the It\^{o} relationship $dw(t)dw(t)^{\top} = (I+\imath \mathbb{J}_m)dt$; see \cite{JNP06}. Note that in this form it follows that $D$ is a real unitary symplectic matrix. That is, it is both unitary (i.e., $DD^{\top}=D^{\top} D=I$) and symplectic (a real $m \times m$ matrix is symplectic if $D \mathbb{J}_m D^{\top} =\mathbb{J}_m$). However, in the most general case, $D$ can be generalized to a symplectic matrix that represents a quantum network that includes ideal infinite bandwidth squeezing devices acting on the incoming field $w(t)$ before interacting with the system \cite{GJN10,NJD08}.  The matrices $A$, $B$, $C$, $D$ of a linear quantum stochastic system cannot be arbitrary and are not independent of one another. In fact, for the system to be physically realizable \cite{JNP06,NJP07b,NJD08}, meaning it represents a meaningful physical system, they must satisfy the constraints (see \cite{WNZJ12,JNP06,NJP07b,NJD08,GJN10})
\begin{eqnarray}
&&A\mathbb{J}_n + \mathbb{J}_n A^{\top} + B\mathbb{J}_mB^{\top}=0, \label{eq:pr-1}\\
&& \mathbb{J}_n  C^{\top} +B\mathbb{J}_mD^{\top}=0, \label{eq:pr-2}\\
&&D \mathbb{J}_m D^{\top} =  \mathbb{J}_{m}. \label{eq:pr-3}
\end{eqnarray}
Note that, more generally, one can consider linear quantum stocastic systems with less outputs than inputs by ignoring certain output quadrature pairs in $y(t)$ which are not of interest, and a corresponding generalized physical realizability conditions analogous to the above can be derived \cite{WNZJ12,Nurd14}. However, for the purpose of this paper it is sufficient to consider systems with the same number of inputs and outputs, as systems with less outputs than inputs can then be easily handled \cite{Nurd14}.

Following \cite{GJ09}, we denote a  linear quantum stochastic system having an equal number of inputs and outputs, and  Hamiltonian $H$,  coupling vector $L$, and scattering matrix $S$, simply as $G=(S, L,H)$ or $G=(S,Kx,\frac{1}{2} x^{\top}Rx)$. We also recall  the  {\em series product} $\triangleleft$  for open Markov quantum systems \cite{GJ09} defined by $G_2 \triangleleft G_1=(S_2S_1, L_2+S_2   L_1,H_1+H_2+\Im\{L_2^{\dag}S_2 L_1\})$, where $G_j=(S_j,L_j,H_j)$ for $j=1,2$. Since the series  product is associative, $G_n \triangleleft G_{n-1} \triangleleft \ldots \triangleleft G_1$ is unambiguously defined. The series product corresponds to a cascade connection $G_2 \triangleleft G_1$ where the outputs of $G_1$ are passed as inputs to $G_2$; see \cite{GJ09} for details.

\subsection{Realization theory}
Given an $n$ degree of freedom linear quantum stochastic  system with system matrices $A,B,C,D$, how can it be built from a bin of linear quantum components and which components are needed? This is the network synthesis question for linear quantum stochastic systems. It was shown in \cite{NJD08} that any linear quantum stochastic system with $n$ degrees of freedom system $G$ can be decomposed as the cascade of $n$ one degree of freedom system $G_1,G_2,\ldots,G_n$ together with some bilinear interaction Hamiltonians  between them, as illustrated in Fig.~\ref{fig:sysdecom}.  It was then shown how each one degree of freedom system can be realized from a certain bin of linear quantum optical components.

\begin{figure}[tbph]
\centering
\includegraphics[scale=0.35]{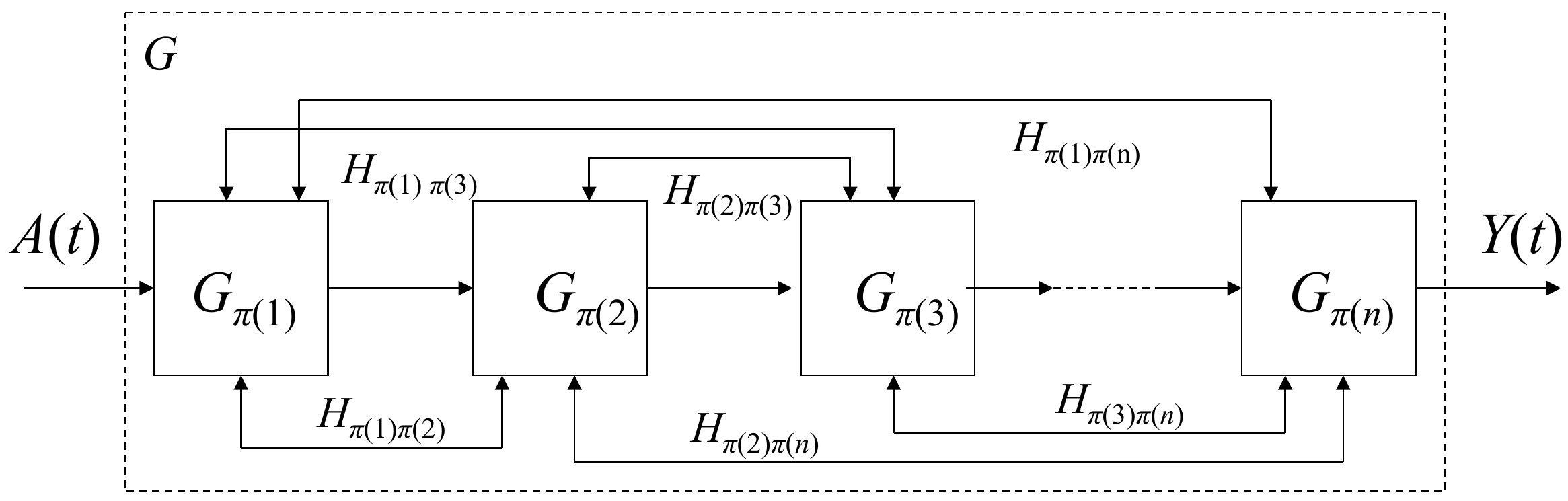}
\caption{Cascade realization of $G_{\pi}$ ($\pi$ is permutation map of $\{1,2,\ldots,n\}$ to itself), with direct interaction Hamiltonians $H^d_{\pi(j)\pi(k)}$ between sub-systems $G_{\pi(j)}$ and $G_{\pi(k)}$ for $j,k=1,2,\ldots,n$, following \cite{NJD08}. Illustration is for $n>3$.}
\label{fig:sysdecom}
\end{figure}

In certain control problems, such as $H^{\infty}$ and $H^2$/LQG coherent feedback control problems, it is the transfer function of the systems that is important rather than the system matrices $G=(A,B,C,D)$ themselves. The transfer function is defined as $\Xi_G(s)=C(sI-A)^{-1}B+D$, and rather than realizing a particular quartet $(A,B,C,D)$ one may consider realizing $G_T=(TAT^{-1},TB,CT^{-1},D)$ for a suitable arbitrary symplectic matrix $T$, since $G$ and $G_T$ have the same transfer function. The transformation $T$ is required to be symplectic to ensure that the new internal degrees of freedom $z(t)=Tx(t)$
satisfies the canonical commutation relations. It was shown in \cite{Nurd10b} that every completely passive linear quantum stochastic system, a system that can be realized using only passive quantum devices, has a pure cascade realization of its transfer function that does require any bilinear interaction Hamiltonians between oscillators in the cascade. In the context of Fig.~\ref{fig:sysdecom} above, this means that  all bilinear interaction Hamiltonians $H_{\pi(i)\pi(j)}$ can be removed.  Purely cascade realizations are simpler to implement and are therefore desirable.
However, it is not known whether the transfer function of general linear quantum stochastic systems outside of the completely passive class have such a realization. This is an important open problem that is resolved in this paper.

\section{A  symplectic QR decomposition algorithm}
\label{sec:sym-GS} The main purpose of this section is to develop a symplectic QR decompositon algorithm and derive a necessary and sufficient condition  for real square matrices of even dimension to possess this decomposition. The symplectic QR decomposition will play an important role in proving subsequent results that will be presented in  Section \ref{sec:pure-cascading}. We begin by recalling some useful definitions.

Let $X$ be an invertible  $2n \times 2n$ skew-symmetric matrix and let $\langle \circ, X \bullet \rangle=\circ^{\top} X \bullet$ be  a skew-symmetric bilinear form on $\mathbb{R}^{2n}$ induced by $X$.  A set of linearly independent vectors $v_1,v_2,\ldots,v_{2n}$ on $\mathbb{R}^{2n}$ is said to be a symplectic basis  with respect to $\langle \circ, X \bullet \rangle$ if $\langle v_i,Xv_j \rangle = -\langle v_j,Xv_i \rangle=X_{ij}$. The space $\mathbb{R}^{2n}$ endowed with  $\langle \circ, X \bullet \rangle$ forms a symplectic vector space. Thus, we shall also refer to $\langle \circ, X \bullet \rangle$ as a symplectic form.  A matrix $T \in \mathbb{R}^{2n \times 2n}$ is said to be sympletic with respect to  $X$ if $T^{\top} X T = X$. Therefore, if $T$ is sympletic with respect to $X$, then $Tv_1,Tv_2,\ldots,Tv_{2n}$ is also a symplectic basis for $\mathbb{R}^{2n}$ whenever $v_1,v_2,\ldots,v_{2n}$ is a symplectic basis. In this paper, we will be interested in $\mathbb{R}^{2n}$ as a symplectic vector space with $X=\mathbb{J}_n$. Therefore, unless stated otherwise, it is implicit throughout that the symplectic structure on $\mathbb{R}^{2n}$ is with respect to the symplectic form $\langle \circ, \mathbb{J}_n \bullet \rangle$. Note the standard property that if $T$ is a symplectic matrix then so is $T^{\top}$ and $T^{-1}$, see, e.g., \cite{deGoss11}. This property will often be invoked without further comment.

We also recall the following definition from \cite{Nurd10b}.
\begin{definition}
A square matrix $F$ of even dimension is said to be lower $2 \times 2$ block triangular if it has a
lower block triangular form when partitioned into $2 \times 2$
blocks:
\begin{eqnarray*}
F= \left[\begin{array}{ccccc} F_{11} & 0_{2 \times 2} & 0_{2
\times 2} & \ldots & 0_{2 \times 2}\\
F_{21} & F_{22} & 0_{2 \times 2} & \ldots & 0_{2 \times 2}\\
\vdots & \ddots & \ddots & \ddots & \vdots \\
F_{n1} & F_{n2} & \ldots & \ldots& F_{nn}
\end{array}\right],
\end{eqnarray*}
where $F_{jk}$, $j\leq k$, is of dimension $2 \times 2$.  Similarly, a matrix $F$ is said to be  upper $2 \times 2$ block triangular  if $F^{\top}$ is lower $2 \times 2$ block triangular.
\end{definition}

We are now ready to state the main lemma of this section that describes a symplectic QR decomposition algorithm with respect to the symplectic form $\langle \circ, \mathbb{J}_n \bullet \rangle$. The lemma is based on a symplectic Gram-Schmidt procedure that is different from symplectic Gram-Schmidt procedures to construct a canonical symplectic basis in a symplectic vector space, e.g., \cite[Proposition 40]{deGoss11} and its proof. It is in the same class of algorithms as, though not identical to, existing symplectic Gram-Schmidt procedures used in numerical analysis with respect to the symplectic form $\langle \circ,  \mathbb{K}_n \bullet \rangle$ with $\mathbb{K}_n = \mathbb{J} \otimes I_n$ \cite{Sal05}. In fact, our procedure is rather analogous to the Gram-Schmidt procedure in spaces with indefinite inner products as described in, e.g. \cite[Section 3.1]{GLR83}. In the lemma, a symplectic basis is constructed sequentially from a {\em given} and {\em fixed} set of linearly independent initial vectors $v_1,v_2,\ldots,v_{2n}$, which are presented to the procedure sequentially two at a time in that order.  As in the Gram-Schmidt procedure in indefinite inner product spaces, since the initial vectors are given, a certain condition is required for the new procedure proposed below to yield a symplectic basis for $\mathbb{R}^{2n}$.

\begin{lemma}[Symplectic QR decomposition]
\label{lem:sym-GS}
Let $V$ be a real invertible $2n \times 2n$ matrix with linearly independent columns $v_1$, $v_2$, $\ldots$, $v_{2n}$ from left to right. Let $M_{i}=\left[\begin{array}{cc} v_{2i-1} & v_{2i} \end{array}\right] \in \mathbb{R}^{2n \times 2}$  for $i=1,2,\ldots,n$
and $\tilde M_1=M_1$, $\tilde M_2= [\begin{array}{cc} M_1 & M_2 \end{array}]$, and $\tilde M_i =[\begin{array}{cccc} M_1 & M_2 & \ldots & M_i \end{array}]$  for $i=3,\ldots,n$,  and assume that $N_i=\tilde M_i^{\top}\mathbb{J}_n \tilde M_i$ is full rank for $i=1,2,\ldots,n$. Then $V$ has a QR decomposition $V=SY$ for some symplectic matrix $S$ and an upper $2 \times 2$ block triangular matrix $Y$. Moreover, $S$ can be constructed recursively by contructing a sequence of real numbers $\alpha_j$,  real $2 (j-1) \times 2$ matrices $\Xi_j$, and  real invertible $2j \times 2j$ matrices $S_{j}$, for $j=2,3,\ldots,n$. Define $\mu_1$ as the (1,2) element of the skew-symmetric matrix $N_1$, $\alpha_1=\sqrt{|\mu_1|}^{-1}$, $\Xi_1=I_2$, and
$$
S_1= M_1 \alpha_1 \Xi_1 \left[\begin{array}{cc} 1 & 0 \\ 0 & {\rm sgn}(\mu_1) \end{array} \right].
$$
For $j=2,3,4,\ldots,n$, define $S_k$ recursively as
\begin{align}
S_j&=\left[\begin{array}{cc} S_{j-1} & M_j\end{array} \right]  \left[\begin{array}{cc}  I_{2(j-1)}  & \alpha_j \Xi_j\left[\begin{array}{cc} 1 & 0 \\ 0 & {\rm sgn}(\mu_j) \end{array} \right]\\ 0 & \alpha_j \left[\begin{array}{cc} 1 & 0 \\ 0 & {\rm sgn}(\mu_j) \end{array} \right] \end{array}\right], \label{eq:def-Sj}
\end{align}
with $\alpha_j  = \sqrt{|\mu_j|}^{-1}$ and $\Xi_j  =\mathbb{J}_{j-1}S_{j-1}^{\top} \mathbb{J}_n M_{j}$,
where $\mu_j$ denotes the (1,2) element of $(S_{j-1} \Xi_j + M_j)^{\top} \mathbb{J}_n (S_{j-1} \Xi _j+ M_j)$.
Then $S_j$ satisfies $S_j^{\top} \mathbb{J}_n S_j = \mathbb{J}_{j}$ for $j=1,2,3,\ldots,n$, and $S= S_n$ is symplectic. In particular, the columns of $S_k$ are contained as the first $2k$ columns of $S_{k+1}$, thus forming a symplectic basis of $\mathbb{R}^{2n}$ for $k=n$. Moreover, defining  the invertible matrices
$$
X_1=\left[\begin{array}{cc} \alpha_1 \Xi_1 \left[\begin{array}{cc} 1 & 0 \\ 0 & {\rm sgn}(\mu_1) \end{array}\right] & 0_{2 \times 2(n-1) } \\ 0_{2(n-1) \times 2} & I_{2(n-1)} \end{array}\right].
$$
and
$$
X_j=\left[ \begin{array}{ccc}   I_{2(j-1)} & \alpha \Xi_j \left[\begin{array}{cc} 1 & 0 \\ 0 & {\rm sgn}(\mu_j) \end{array}\right]& 0 \\ 0 & \alpha_j \left[\begin{array}{cc} 1 & 0 \\ 0 & {\rm sgn}(\mu_j) \end{array}\right] & 0 \\ 0 & 0 & I_{2(n-j)} \end{array} \right]
$$
for $j=2,3,\ldots,n$, then $V= SY$ for an invertible upper $2 \times 2$ block triangular matrix $
Y=X^{-1}$, where $X=X_1  X_2 \cdots  X_{n-1}X_n$.
\end{lemma}
\begin{proof}
Since $M_1^{\top} \mathbb{J}_n M_1$ is a real $2 \times 2$ skew-symmetric matrix and is full rank by hypothesis,  it is of the form $$M_1^{\top} \mathbb{J}_n M_1 = \left[\begin{array}{cc} 0 & \mu_1 \\ -\mu_1 & 0 \end{array}\right],$$ with $\mu_1 \neq 0$. It follows immediately from the given construction of $S_1$ in the statement of the theorem that $S_1^{\top}  \mathbb{J}_n S_1=\mathbb{J}_1$. Therefore, the columns of $S_1$ are mutually skew-orthogonal.  We proceed further by induction.

Suppose that  the columns of $S_j$, as constructed according to the theorem, form a partial symplectic basis for $1<j < n$, i.e., $S_j^{\top} \mathbb{J}_n S_j = \mathbb{J}_j$. Consider now the matrix $Z_{j+1} = S_j \Xi_{j+1} + M_{j+1}$ for some real $2j \times 2$ matrix $\Xi_{j+1}$. We will choose $\Xi_{j+1}$ to satisfy $S_{j}^{\top} \mathbb{J}_n Z_{j+1}=0$. This yields the equation $S_j^{\top} \mathbb{J}_n S_j \Xi_{j+1} + S_j^{\top} \mathbb{J}_n M_{j+1} =0$. Since $S_j^{\top} \mathbb{J}_n S_j = \mathbb{J}_{j}$ we can solve for $\Xi_{j+1}$ to obtain $\Xi_{j+1} = \mathbb{J}_j S_j^{\top} \mathbb{J}_n M_{j+1}$.
Define the real $2 \times 2$ skew-symmetric matrix $Y_{j+1}=Z_{j+1}^{\top} \mathbb{J}_n Z_{j+1}$. We will  show that $Y_{j+1} \neq 0_{2 \times 2}$. We first note that since $N_{j+1}=[\begin{array}{cccc} M_1 & \ldots & M_j & M_{j+1} \end{array}]^{\top} \mathbb{J}_n[\begin{array}{cccc} M_1 & \ldots & M_j & M_{j+1} \end{array}]$ is full rank by hypothesis,  so is $[\begin{array}{cc} S_j  & M_{j+1} \end{array}]^{\top} \mathbb{J}_n[\begin{array}{cc} S_j & M_{j+1} \end{array}]$. This is a consequence of the fact that the columns of $S_j$ are, by construction, linearly independent linear combinations of the columns of $M_1,M_2,\ldots, M_j$. Moreover, since
$$
[\begin{array}{cc} S_j  & Z_{j+1} \end{array}]=[\begin{array}{cc} S_j  & M_{j+1} \end{array}]\left[\begin{array}{cc} I_j & \Xi_{j+1} \\ 0 & I_2\end{array}\right],
$$
while the matrix
$$
\left[\begin{array}{cc} I_j & \Xi_{j+1} \\ 0 & I_2\end{array}\right]
$$
is evidently invertible, we conclude that $$[\begin{array}{cc} S_j  & Z_{j+1} \end{array}]^{\top} \mathbb{J}_n [\begin{array}{cc} S_j  & Z_{j+1} \end{array}]$$ is full rank skew-symmetric. By the given construction of $S_j$ and $\Xi_{j+1}$, $[\begin{array}{cc} S_j  & Z_{j+1} \end{array}]^{\top} \mathbb{J}_n [\begin{array}{cc} S_j  & Z_{j+1} \end{array}]$ is necessarily of the form
$$
[\begin{array}{cc} S_j  & Z_{j+1} \end{array}]^{\top} \mathbb{J}_n [\begin{array}{cc} S_j  & Z_{j+1} \end{array}] =\left[\begin{array}{cc} \mathbb{J}_j & 0 \\ 0 & Y_{j+1} \end{array}\right].
$$
From the fact that the left hand side of the identity is full rank, it follows immediately that $Y_{j+1}$ is  full rank $2 \times 2$ skew-symmetric. Therefore, it is necessarily of the form
$$
Y_{j+1}=\left[\begin{array}{cc} 0 & \mu_{j+1} \\ -\mu_{j+1} & 0\end{array} \right],
$$
with $\mu_{j+1} \neq 0$.  Define $\alpha_{j+1} = \sqrt{|\mu_{j+1}|}^{-1}$. Consider now the matrix
$$
\tilde Z_{j+1}=\alpha_{j+1} Z_{j+1} \left[\begin{array}{cc} 1 & 0 \\ 0 & {\rm sgn}(\mu_{j+1})\end{array}\right].
$$
Some brief calculations shows that, by construction, the matrix $\tilde Z_{j+1}$ satisfies $ S_j^{\top}  \mathbb{J}_n \tilde Z_{j+1}=0$, and
$
\tilde Z_{j+1}^{\top} \mathbb{J}_n \tilde Z_j =\mathbb{J}_1.
$
Therefore, we have shown for $1 <j<n$ that if the hypotheses of the theorem hold  and $S_j$ satisfies $S_j^{\top} \mathbb{J}_n S_j=\mathbb{J}_j$ then the matrix $S_{j+1}$ given by (\ref{eq:def-Sj})  satisfies $S_{j+1}^{\top} \mathbb{J}_n S_{j+1}=\mathbb{J}_{j+1}$. In particular, $S=S_n$ is a symplectic matrix.

Note that by the above construction each  $X_i$ as defined in the lemma is invertible. Direct calculations then show that $VX_1=[\begin{array}{cc} S_1 & V_{2\rightarrow n} \end{array}]$, $VX_1 X_2 =[\begin{array}{cc} S_2 & V_{3 \rightarrow n} \end{array}]$, $\ldots$, $VX_1X_2 \cdots X_n=S_n = S$, where $V_{j \rightarrow n}$ is matrix constructed of columns  $2j$ to $2n$ of $V$ from left to right. Moreover, clearly $X=X_1X_2 \cdots X_n$ is upper $2\times 2$ block triangular since each of the $X_j$ in the product has this structure, and therefore so is $Y=X^{-1}$. Hence, $V$ has the symplectic QR decomposition $V=SX^{-1}=SY$. This concludes the proof of the lemma. 
\end{proof}

Let us now look at an example to illustrate Lemma \ref{lem:sym-GS}.
\begin{example}
Consider the matrix
$$
V=\left[\begin{array}{cccc}
-15 & 42 & -12 & 3 \\
33 & -22 & 7 & 28 \\
9 & 26 & -43 & 44 \\
5 & 26 & -45 & -37
\end{array}\right].
$$
It can be verified that $V^{\top} \mathbb{J}_2 V$ and the matrix $N_1$ as defined in Lemma \ref{lem:sym-GS} are full rank. The symplectic matrix produced by executing the symplectic QR decomposition is then
$$
S=\left[\begin{array}{cccc}
-0.4862 & -1.3612 & -0.1418 & -1.0405 \\
1.0695 & 0.7130 & 0.0975 & 1.4863 \\
0.2917 & -0.8427 & -0.7193 & 0.4113 \\
0.1621 & -0.8427 & -0.7569 & -1.1093
\end{array}\right],
$$
while the upper $2\times2$ block triangular matrix $Y$ such that $V=SY$ is
$$
Y=\left[\begin{array}{cccc}
30.8545    &     0  &  -0.7130  & -28.0024 \\
         0 &  -30.8545    & 3.2734  & -34.7437 \\
         0     &    0   &  55.6558     &    0 \\
         0    &     0    &     0  &  55.6558
\end{array}\right];
$$
\end{example}

Finally, we give  a necessary and sufficient condition for a $2n \times 2n$ matrix $V$ to possess a symplectic QR decomposition.

\begin{theorem}
\label{thm:iff-sym-GS} Let $V$ and $N_j$ be as defined in Lemma \ref{lem:sym-GS}. Then there exists a symplectic $QR$ decomposition $V=SY$, with $S$ a symplectic matrix and $Y$ an invertible upper $2\times 2$ block triangular matrix, if and only if the matrices $N_1$, $N_2$, $\ldots$, $N_{n}$ are full rank.
\end{theorem}
\begin{proof}
The if part is the content of Lemma  \ref{lem:sym-GS}. For the necessity of the full rankness of $N_1$, $N_2$, $\ldots$, $N_{n-1}$,  first note that $V^{\top} \mathbb{J}_n V = Y^{\top}S^{\top} \mathbb{J}_n  SY = Y^{\top} \mathbb{J}_n  Y$. Take any $k \in \{1,2,\ldots,n-1\}$
and partition $Y$ as
$$
Y =\left[\begin{array}{cc} Y_{11,k} & Y_{12,k} \\ 0 & Y_{22,k}\end{array} \right],
$$
with $Y_{11,k} \in \mathbb{R}^{2k \times 2k}$, $Y_{12,k} \in \mathbb{R}^{2k \times 2(n-k)}$, and $Y_{22,k} \in \mathbb{R}^{2(n-k) \times 2(n-k)}$, where $Y_{11,k}$ and $Y_{22,k}$ are invertible upper $2 \times 2$ block triangular matrices. Since $N_k = [\begin{array}{cc} I_{2k} & 0 \end{array}] V^{\top} \mathbb{J}_n V [\begin{array}{cc} I_{2k} & 0 \end{array}]^{\top}$, from the expression for $Y$ above we immediately get that $N_k =Y_{11,k}^{\top} \mathbb{J}_k Y_{11,k}$, which is evidently invertible for $k=1,2,\ldots,n-1$, while $N_n=Y^{\top} \mathbb{J}_n  Y$ is invertible by hypothesis. 
\end{proof}

We now provide an example of an instance  where the condition of Theorem \ref{thm:iff-sym-GS}  {\em fails}  and hence   $V$ does not have symplectic QR decompositon.
\begin{example}
\label{eg:rank-fail} Consider the matrix
$$
V=\left[\begin{array}{cccc}
0 & 1 & 0 & 0 \\
0 & 0 & 0 & 1 \\
-1 & 0 & 0 & 0 \\
0 & 0 & -1 & 0
\end{array}
\right].
$$
Then $V^{\top} \mathbb{J}_n V$ is full rank. However, it may be verified that the matrix $N_1$  associated with $V$ is a zero matrix, hence the condition of Theorem \ref{thm:iff-sym-GS} is not satisfied and $V$ does not have a symplectic QR decomposition.
\end{example}

\section{Pure cascade realization of the transfer function of  generic linear quantum stochastic systems}
\label{sec:pure-cascading} In this section we employ the results from the preceding section to obtain sufficient conditions under which a (physically realizable) linear quantum stochastic system has a pure cascade realization. It is shown that this condition will be met by  generic (in a sense that will be detailed in this section) linear quantum stochastic systems, thereby extending the results obtained in \cite{Nurd10b}  for the special case of completely passive systems to generic linear quantum stochastic systems (including a large class of active systems). Let us now recall a characterization of  linear quantum stochastic systems $G$ that have a pure cascade realization, i.e., $G$ can be written as $G=G_n \triangleleft G_{n-1} \triangleleft \cdots \triangleleft G_1$ for some distinct one degree of freedom systems $G_1,G_2,\ldots,G_n$.

\begin{theorem}
\label{th:casc-struc} \textbf{\cite[Theorem 4]{Nurd10b}} Let $R=[R_{ij}]_{i,j=1,2,\ldots,n}$ with $R_{ij} \in \mathbb{R}^{2 \times 2}$, and $K=[\begin{array}{cccc} K_1 & K_2 & \ldots & K_n \end{array}]$ with $K_i \in \mathbb{C}^m$. A linear quantum stochastic system
$G=(S,Kx,\frac{1}{2}x^{\top} R x)$ with $n$ degrees of freedom is
realizable by a pure cascade of $n$ one degree of freedom harmonic
oscillators (without a direct interaction Hamiltonian) if and only if the $A$
matrix is similar via a symplectic permutation matrix to a lower $2 \times 2$ block triangular matrix. That is, there exists a symplectic permutation matrix $P$ such that $PAP^{\top}= \tilde A$, where $\tilde A$ is lower $2  \times 2$ block triangular. Let $\tilde R=PRP^{\top}=[\tilde R_{ij}]$, $\tilde K = KP^{\top} = [\begin{array}{cccc} \tilde K_1 & \tilde K_2 & \ldots & \tilde K_n \end{array}]$, with $\tilde R_{ij} \in \mathbb{R}^{2 \times 2}$ and $\tilde K_j \in \mathbb{C}^{m}$. If the condition is satisfied then $G$ can be explicitly constructed as
the cascade connection $G_n \triangleleft G_{n-1} \triangleleft
\ldots \triangleleft G_1$  with $G_1=(S,\tilde K_1 \tilde x_1,\frac{1}{2} \tilde x_1^{\top}
\tilde R_{11}\tilde x_1)$, and $G_k=(I,\tilde K_k \tilde x_k,\frac{1}{2} \tilde x_k^{\top} \tilde R_{kk} \tilde x_k)$ for
$k=2,\ldots,n$, where
$\tilde x =(q_{\pi(1)},p_{\pi(1)},q_{\pi(2)},p_{\pi(2)},\ldots,q_{\pi(n)},p_{\pi(n)})^{\top}$, and $\pi$ is a permutation of $\{1,2,\ldots,n\}$ onto itself such that $\tilde x = Px$.
\end{theorem}

\begin{remark}
The theorem has been stated as a minor and trivial generalization of \textbf{\cite[Theorem 4]{Nurd10b}}. The original did not include the additional freedom of allowing a symplectic permutation matrix to transform $A$ into lower $2 \times 2$ block triangular form, corresponding to a mere permutation of pairs of position and momentum operators in $x$. For instance, if $A$ is in upper $2 \times 2$ block triangular form it can be trivially transformed into lower  $2 \times 2$ block triangular form by a suitable symplectic permutation matrix, which by \textbf{\cite[Theorem 4]{Nurd10b}} would then be physically realizable by pure cascading.
\end{remark}

Given an $n$ degree of freedom linear quantum stochastic system $G=(A,B,C,D)$ with transfer function $\Xi_G(s)=C(sI-A)^{-1}B+D$, the problem that will be addressed  is how to obtain a cascade of $n$ one degree of freedom linear quantum stochastic systems that has transfer function $\Xi_G(s)$, if such a cascade exists. Recall that for any symplectic matrix $T$, the system $G_T=(TAT^{-1},TB,CT^{-1},D)$  is a physically realizable system that has the same transfer function as $G$. The main strategy is to find a symplectic matrix $T$ such that $G_T$ is the cascade realization that is sought. Before stating the main result, let us recall the real Jordan canonical form  of a real matrix; see, e.g., \cite[Section 3.4]{HJ85}. Let $A$ be a real square matrix then $A$ can always be decomposed as $A=  V  J_A  V^{-1}$, where $ J_A$ is a Jordan canonical form for $A$. Of course, although $A$ is real, its eigenvalues and eigenvectors can be complex,  but they always come in complex conjugate pairs. That is, if $\lambda$ and $v$ are a complex eigenvalue and eigenvector of $A$ then so are $\lambda^*$ and $v^{\#}$, respectively. Therefore, in general, $V$ and $ J_A$ may have complex entries. However, when $A$ is real it is also similar to a real Jordan canonical form. A real Jordan  block for a real Jordan canonical corresponding to a real eigenvalue of $A$ is the same as the corresponding block in the Jordan canonical form. However,  to a pair of conjugate complex eigenvalues $\lambda_k=a_k+\imath b_k $ and $\lambda^*=a_k-\imath b_k $ there will associated with them one or more real Jordan blocks of the upper $2 \times 2$ block triangular form
$$
\left[\begin{array}{cccccc}  C_k & I_2 & 0_{2 \times 2} & 0_{2 \times 2} & \ldots & 0_{2 \times 2} \\
0_{2 \times 2} & C_k & I_2 & 0_{2 \times 2} & \ldots & 0_{2 \times 2}  \\
0_{2 \times 2} & 0_{2 \times 2} & C_k & \ddots & \ddots & 0_{2 \times 2} \\
\vdots & \vdots & \vdots & \ddots & \ddots & I_2 \\
0_{2 \times 2} & 0_{2 \times 2} &  0_{2\times 2 } & 0_{2 \times 2} & \ldots & C_k\end{array}\right],
$$
with
$$
C_k=\left[\begin{array}{cc} a_k & b_k \\ -b_k & a_k \end{array} \right].
$$
With respect to the real Jordan blocks, $A$ can be written as $\tilde V \tilde J_A \tilde V^{-1}$, where $\tilde J_A$ is a real Jordan canonical form of $A$ (unique up to permutation of the real Jordan blocks), and $\tilde V$ a real invertible matrix \cite[Section 3.4]{HJ85}. We are now ready to state the main results of this section.

\begin{lemma}[Symplectic Schur decomposition]
\label{lem:sym-Schur-decomposition} Let $A$ be a real $2n \times 2n$ matrix. Then $A$ has a symplectic Schur decomposition  $A= S^{-1}US$ with $U$ lower $2 \times 2$ block triangular if there exists a real invertible $2n \times 2n$ matrix
$$
\tilde V =[\begin{array}{ccccccc}  \tilde v_1 &  \tilde v_2 & \ldots &  \tilde v_{2n-1} & \tilde v_{2n} \end{array}]
$$
such that
\begin{description}
\item[(i)] $\tilde V$ brings $A$ into a real Jordan canonical form $\tilde J_A= \tilde V^{-1} A \tilde V$, with $\tilde J_A$ in the upper $2 \times 2$ block triangular form $\tilde J_A= {\rm diag}(\tilde J_{A,r}, \tilde J_{A,c})$, where $\tilde J_{A,r}$ contains all  Jordan blocks  corresponding to the (possibly repeated) real eigenvalues of $A$ (in upper triangular form), and  $\tilde J_{A,c}$ contains all real Jordan blocks corresponding to the complex eigenvalues of $A$.

\item[(ii)] The matrices $\tilde N_1$, $\tilde N_2$, $\ldots$, $\tilde N_{n-1}$ given by
\begin{eqnarray*}
\tilde N_1 &=&  [\begin{array}{cc} \tilde v_1 & \tilde v_2 \end{array}]^{\top} \mathbb{J}_n [\begin{array}{cc} \tilde v_1 & \tilde v_2 \end{array}],\\
\tilde N_2 &=& [\begin{array}{cccc} \tilde v_1 & \tilde v_2 & \tilde v_3 & \tilde v_4 \end{array}]^{\top} \mathbb{J}_n [\begin{array}{cccc} \tilde v_1 & \tilde v_2 & \tilde v_3 & \tilde v_4 \end{array}],\\
&\vdots& \\
\tilde N_{n-1} &=& [\begin{array}{ccccc} \tilde v_1 & \tilde v_2 & \ldots & \tilde v_{2n-3} & \tilde v_{2n-2} \end{array}]^{\top} 
\mathbb{J}_n [\begin{array}{ccccc} \tilde v_1 & \tilde v_2 & \ldots & \tilde v_{2n-3} & \tilde v_{2n-2} \end{array}].
\end{eqnarray*}
are all full rank.

If the above conditions hold, let $S_1$ be a $2n \times 2n$ symplectic matrix obtained from $\tilde V$ by applying the symplectic QR decomposition of Lemma \ref{lem:sym-GS} so that $\tilde V= S_1 Y$, for an invertible  $2n \times 2n$ upper $2 \times 2$ block triangular matrix $Y$ as given by the lemma, and let $P \in \mathbb{R}^{2n \times 2n}$ be a permutation matrix that implements the mapping $(q_1,p_1,q_2,p_2,\ldots,q_n,p_n)^{\top} \mapsto (q_n,p_n,q_{n-1},p_{n-1},\ldots,q_1,p_1)^{\top}$. Then $A$ has the symplectic Schur decomposition $A=S^{-1} U S$ with $S=PS_1^{-1}$ and $U$ lower $2\times 2$ block triangular.
\end{description}
\end{lemma}

\begin{remark}
Notice that since $A$ and $\tilde J_{A,c}$ have even dimensions, so does $\tilde J_{A,r}$. One can always choose a real Jordan canonical form of $A$ to be of the form $\tilde J_A= {\rm diag}(\tilde J_{A,r}, \tilde J_{A,c})$, which is upper $2 \times2$ block triangular.
\end{remark}

\begin{proof}
Let $A= \tilde V \tilde J_A \tilde V^{-1}$, with $\tilde J_A$ and $\tilde V$ as given in the lemma. By conditions (i) and (ii), using Lemma \ref{lem:sym-GS} we can construct a symplectic matrix $S_1$ such that $\tilde V =S_1 Y$, with $Y$  real invertible upper $2 \times 2$ block triangular as given in the lemma. We can thus  write $A= \tilde V  \tilde J_A \tilde V^{-1} = S_1 Y  \tilde J_A Y^{-1} S_1^{-1}$.
Moreover, since $\tilde J_A= {\rm diag}(\tilde J_{A,r}, \tilde J_{A,c})$, $Y$, and $Y^{-1}$ are all upper $2 \times 2$ block triangular, the product $Y \tilde J_A Y^{-1}$ is also upper $2 \times 2$ block triangular.  Let $P \in \mathbb{R}^{2n \times 2n}$ be the  permutation matrix defined in the lemma.  Notice that, by its definition, the permutation matrix $P$ is symplectic, and that $P Z P^{\top}$ is $2 \times 2$ {\em lower} block triangular whenever $Z$ is  {\em upper} $2\times 2$ block triangular. It follows from these observations that  $PY \tilde J_A Y^{-1}P^{\top}$ is lower $2\times 2$ block triangular. Therefore, we conclude that $SAS^{-1}$ with $S=PS_1^{-1}$ is a lower $2\times 2$ block triangular matrix, since  $SAS^{-1}= PY \tilde J_A Y^{-1} P^{\top}$.  Additionally, notice that $S$ is symplectic since $S_1^{-1}$ (the inverse of a symplectic matrix) and $P$ are symplectic. $\hfill \Box$
\end{proof}

A direct consequence of Lemma \ref{lem:sym-Schur-decomposition} is the existence of the cascade realization of the transfer function of a linear quantum stochastic system when the conditions of the lemma are satisfied.
\begin{theorem}
\label{thm:cascade-real} Let $G=(A,B,C,D)$ be a physically realizable $n$ degree of freedom linear quantum stochastic system. If there exists a matrix $\tilde V$ associated to the $2n \times 2n$ matrix $A$ satisfying the conditions of Lemma \ref{lem:sym-Schur-decomposition} then there exists  a symplectic matrix $S$ such that the transformed system $(SAS^{-1}, SB, CS^{-1},D)$ is physically realizable with $SAS^{-1}$  lower $2 \times 2$ block triangular, i.e., $\Xi_G(s)$ has a pure cascade realization.
\end{theorem}

\begin{proof}
Let $S$ be as in  Lemma \ref{lem:sym-Schur-decomposition}, then $A=SAS^{-1}$ is lower $2 \times 2$ block triangular. Therefore, $\Xi_G(s)$ has a pure cascade realization by Theorem \ref{th:casc-struc}. $\hfill \Box$
\end{proof}

We emphasize that fulfillment of  the full rankness conditions on $\tilde N_1$, $\tilde N_2$, $\ldots$, $\tilde N_{n-1}$ depends on the choice of the matrix $\tilde V$ which transforms  $A$  into its real Jordan canonical form (which is not unique). For some choices of $\tilde V$ the full rankness conditions may fail to hold and thus a pure cascade realization of the transfer function cannot be obtained.

Let us call as {\em admissible} all real $2n \times 2n$ matrices $A$ satisfying the conditions of Lemma  \ref{lem:sym-Schur-decomposition}, and refer to those that do not as {\em non-admissible}. The following examples illustrate some samples of  non-admissible matrices that cannot meet  the conditions of  Lemma  \ref{lem:sym-Schur-decomposition}.

\begin{example}
\label{ex:non-admissible-1} Consider the matrix
$$
A=\left[\begin{array}{cccc}   -1 & 0 & 0 & -1 \\  0 & -1 & 0 & 0 \\ -1 & 0 & -1 & 0 \\ 0 & -1 & 0 & -1 \end{array}\right],
$$
which has the real Jordan decomposition $A=\tilde V \tilde J_A \tilde V^{-1}$ with (following from Example \ref{eg:rank-fail})
$$
\tilde V =\left[\begin{array}{cccc} 0 & 1 & 0 & 0 \\ 0 & 0 & 0 & 1 \\ -1 & 0 & 0 & 0 \\  0 & 0 & -1 & 0 \end{array}  \right],\;
\tilde J_A = \left[\begin{array}{cccc} -1 & 1 & 0 & 0 \\ 0 & -1 & 1 & 0 \\ 0 & 0 & -1 & 1 \\ 0 & 0 & 0 & -1 \end{array} \right].
$$
It can be easily inspected that for any choice of permutation matrix $P$ such that $ \tilde V P$ satisfies the conditions of Theorem \ref{thm:iff-sym-GS}, one will find that $P \tilde J_A P^{\top}$ will not be upper $2 \times 2$  block triangular. This matrix $A$ is therefore non-admissible.
\end{example}

\begin{example}
\label{ex:non-admissible-2} Consider the matrix
$$
A=\left[\begin{array}{cccc}   -2 & 0 & 0 & 0 \\  0 & -3 & 0 & 4 \\ 0 & 0 & -1 & 0 \\ 0 & -4 & 0 & -3 \end{array}\right],
$$
which has the real Jordan decomposition $A=\tilde V \tilde J_A \tilde V^{-1}$ with
$$
\tilde V =\left[\begin{array}{cccc} 0 & 1 & 0 & 0 \\ 0 & 0 & 0 & 1 \\ -1 & 0 & 0 & 0 \\  0 & 0 & -1 & 0 \end{array}  \right],\;
\tilde J_A = \left[\begin{array}{cccc} -1 & 0 & 0 & 0 \\ 0 & -2 & 0 & 0 \\ 0 & 0 & -3 & 4 \\ 0 & 0 & -4 & -3 \end{array} \right].
$$
As with Example \ref{ex:non-admissible-2} it can be verified that  for any choice of permutation matrix $P$ such that $ \tilde V P$ satisfies the conditions of Theorem \ref{thm:iff-sym-GS}, one will find that $P \tilde J_A P^{\top}$ will not be upper $2 \times 2$  block triangular. Thus $A$ is also non-admissible.
\end{example}

We observe the following:
\begin{enumerate}
\item All  diagonalizable matrices in $\mathbb{R}^{2n \times 2n}$ (including all symmetric matrices) with only real eigenvalues are admissible. Example  \ref{ex:NOPA-cascade} to be given below involves this type of admissible matrix.

\item Non-admissible matrices in $\mathbb{R}^{2n \times 2n}$ include the following cases:
\begin{description}
\item[(i)] The matrix has a real repeated eigenvalue with geometric multiplicity less than its algebraic multiplicity, and there exist two mutually skew-orthogonal real basis vectors $v_1$ and $v_2$ for the invariant subspace of $\mathbb{R}^{2n}$ associated with that eigenvalue. For some matrices  of this type it is not possible to permute the columns of $\tilde V$ and the corresponding rows and columns of $\tilde J_A$ to transform them into  admissible matrices, as illustrated in  Example \ref{ex:non-admissible-1}.

\item[(ii)] The matrix has a pair of conjugate complex eigenvalues $\lambda$ and $\lambda^*$ (not necessarily repeated) with a corresponding pair of  conjugate eigenvectors or generalized eigenvectors $v$ and $v^{\#}$ such that the real vectors $v+v^{\#}$ and $-\imath v + \imath v^{\#}$ are mutually skew-orthogonal.  Example \ref{ex:non-admissible-2} illustrates an instance of a non-admissible matrix with this property.
\end{description}
\end{enumerate}

The non-admissibility of a real $2n \times 2n$ matrix  entails rather particular properties  that are unlikely to be possessed by typical matrices. This suggests that admissible real $2n \times 2n$ matrices are generic in the set of all real $2n \times 2n$ matrices. Generic is in the sense that the set of admissible matrices contains an open and dense subset of $\mathbb{R}^{2n \times 2n}$. This is indeed the case  and we state it as the following theorem, with the proof being deferred to the appendix.

\begin{theorem}
\label{thm:admissibility-is-generic} The set of admissible $2n \times 2n$ matrices is generic in  $\mathbb{R}^{2n \times 2n}$.
\end{theorem}

Thus,   generic  matrices in $\mathbb{R}^{2n \times 2n}$  have a symplectic Schur decomposition and the transfer function $\Xi_G(s)$ of a generic physically realizable linear quantum stochastic system  $G=(A,B,C,D)$ has a pure cascade realization that can be  explicitly determined using Lemma \ref{lem:sym-Schur-decomposition} and Theorem \ref{thm:cascade-real}. We conclude this section by applying the results obtained herein in an example that demonstrates an equivalent realization of the transfer function of a nondegenerate optical parametric amplifier (NOPA) by a cascade of two  degenerate parametric amplifiers (DPAs) equipped with an additional transmissive mirror.

\begin{example}
\label{ex:NOPA-cascade} Consider a NOPA with two modes $a_j = \frac{1}{2}(q_j+\imath p_j)$, $j=1,2$, satisfying the canonical commutation relations $[a_j,a_k^*]=\delta_{jk}$ and $[a_j,a_k]=0$. The operators describing the system is $H=\frac{\imath \epsilon}{2}(a_1^*a_2^*-a_1 a_2)$, $L=[\begin{array}{cc} \sqrt{\gamma} a_1 & \sqrt{\gamma} a_2 \end{array}]^{\top}$, and $S=I_2$. We take $\gamma=7.2 \times 10^7$ and $\epsilon = 0.6 \gamma = 4.32 \times 10^7$, values that can be realized  in a tabletop optical experiment, see, e.g., the experimental work \cite{IYYYF11} based on the proposals in \cite{YK03b,GW09}. The $A,B,C,D$ matrices for the NOPA are:
\begin{eqnarray*}
A &=& 10^7 \left[\begin{array}{cccc} -3.6 & 0 & 2.16 & 0 \\
0 & -3.6 & 0 & -2.16 \\
2.16 & 0 & -3.6 & 0\\
0 & -2.16 & 0 & -3.6
\end{array}\right]; \\
B&=&-8.4853 \times 10^3 I_4;\;C = 8.4853 \times 10^3 I_4;\; D= I_4.
\end{eqnarray*}
We can choose the matrix $\tilde V$ in Theorem \ref{thm:cascade-real} to be
$$
\tilde V = \left[\begin{array}{cccc} 0.7071 & 0 & 0  &  0.7071 \\
0 & -0.7071 & 0.7071 & 0\\
-0.7071 & 0 & 0 & 0.7071\\
0 & 0.7071 & 0.7071 & 0 \end{array}\right],
$$
corresponding to the (real) Jordan canonical form
$$
 \tilde J_A = \tilde V^{-1} A \tilde V = 10^7 {\rm diag}(-5.76,-1.44,-5.76,-1.44).
$$
Using Lemma \ref{lem:sym-GS}, we compute the symplectic matrix  $S_1$ and upper $2 \times 2$ block triangular matrix $Y$ as
\begin{eqnarray*}
S_1 &=& \left[\begin{array}{cccc} 0.7071 & 0 & 0 & -0.7071 \\
0 & 0.7071 &  0.7071 & 0 \\
-0.7071 & 0 & 0 & -0.7071 \\
0 & -0.7071 & 0.7071 & 0
\end{array}\right]; \\
Y&=& {\rm diag}(1,-1,1,-1).
\end{eqnarray*}

The required symplectic transformation matrix from Theorem \ref{thm:cascade-real} is $S=PS_1^{-1}$, and a cascade realization of the transfer function of the NOPA is $G_1=(A_1,B_1,C_1,D_1)=(SAS^{-1},SB,CS^{-1},D)$ with
\begin{eqnarray*}
A_1 &=& 10^7 \left[\begin{array}{cccc} -5.76 & 0 & 0 & 0 \\
0 & -1.44 & 0 & 0 \\
0 & 0 & -5.76 & 0\\
0 & 0 & 0 & -1.44
\end{array}\right]; \\
B_1&=&10^3 \left[\begin{array}{cccc}  0 & -6 & 0 & -6 \\
6 & 0 & 6 & 0 \\
-6 & 0 & 6 & 0\\
0 & -6 & 0 & 6
\end{array}\right]; \;
C_1 = 10^3 \left[\begin{array}{cccc} 0 & -6 & 6 & 0 \\
6 & 0 & 0 & 6 \\
0 & -6 & -6 & 0\\
6 & 0 & 0 & -6
\end{array}\right]; \\
D_1&=& I_4.
\end{eqnarray*}
The cascade realization $G_1$ can be decomposed as the cascade $G_1 =G_{12} \triangleleft G_{11}$, with
$$
G_{11}=\left(I_2, 10^3 \left[\begin{array}{cc} 3\imath & -3 \\ 3\imath & -3\end{array} \right]\left[\begin{array}{c} q_1 \\ p_1 \end{array} \right],-5.4 \times 10^6 (q_1p_1+p_1 q_1) \right),$$
and
$$
G_{12}=\left(I_2, 10^3 \left[\begin{array}{cc} 3 & 3\imath \\ -3 & -3\imath \end{array} \right]\left[\begin{array}{c} q_2 \\ p_2 \end{array} \right],-5.4  \times 10^6 (q_2p_2+p_2 q_2) \right).$$
Each of $G_{11}$ and $G_{12}$ can be realized  as a DPA with  two transmissive mirrors rather than one; see \cite{NJD08} for details of the realization of $G_{11}$ and $G_{12}$. Note that the pump amplitude for each NOPA in the cascade realization is $4 \times 5.4 \times 10^6 = 2.16\times 10^7$. Therefore, remarkably,  the cascade realization $G_{12} \triangleleft G_{12}$ requires less total pump power to realize than the original NOPA, i.e., $2 \times (2.16\times 10^7)^2$ in the cascade compared  to $(4.32 \times 10^7)^2$ in the original, i.e., half the pump power. So, with the cascade realization one obtains a more power efficient realization of the same transfer function which yields the same amount of two-mode squeezing in the two output beams.

Finally, note that if $\tilde V$ had been chosen differently from the one  above, for instance, as
$$
\tilde V = \left[\begin{array}{cccc} 0.7071 & 0 & 0.7071  & 0 \\
0 & 0.7071 & 0 & -0.7071 \\
-0.7071 & 0 & 0.7071 &  0\\
0 & 0.7071 & 0 & 0.7071 \end{array}\right],
$$
corresponding to
$$
 \tilde J_A = \tilde V^{-1} A \tilde V = 10^7 {\rm diag}(-5.76,-5.76,-1.44,-1.44),
$$
then it may be readily inspected that the full rankness conditions of Lemma \ref{lem:sym-Schur-decomposition} are not satisfied, hence  this choice of $\tilde V$ cannot lead to a pure cascade realization of the NOPA.
\end{example}

\section{Conclusion}
\label{sec:conclu} In this paper we have generalized the ideas and results in \cite{Pet12,Nurd10b}, that focus on the special class of completely passive linear quantum stochastic systems, to show that the transfer function of generic linear quantum stochastic systems, which includes a large generic class of active systems, can be realized by pure cascading. The proof is constructive as the cascade realization, when it exists, can be explicitly computed.  This is of practical importance as it will allow a simpler realization of a large class of linear quantum stochastic systems as, say, coherent feedback controllers or quantum optical filters. Numerical examples have been provided to illustrate the results of the paper. In one example, it is shown that the transfer function of a nondegenerate optical parametric amplifier has a realization as the cascade of two degenerate optical parametric amplifiers having an additional outcoupling mirror, which operates for only  half of the pump power required by the nondegenerate optical parametric amplifier.

\vspace{0.5cm}

\noindent \textbf{Acknowledgement.} Contributions: HN developed the symplectic QR and Schur decomposition algorithms, associated results and Example \ref{ex:NOPA-cascade}, SG and IP proved the genericity of admissible matrices in discussion with HN. The authors thank the reviewers and Associate Editor for their constructive and helpful comments  on this paper.

\appendices
 
\section*{Proof of Theorem \ref{thm:admissibility-is-generic}}

\newcommand{\diag}{\ensuremath{\mathrm{diag}}}

Let $M_{2n}(\mathbb{R})$ denote the set of $2n \times 2n$ real matrices, $\bar{M}_{2n}(\mathbb{R})$ denote the subset of full rank (i.e. invertible) matrices, and $\bar{M}_{2n,s}(\mathbb{R})$ the subset of those matrices in $\bar{M}_{2n}(\mathbb{R})$ that are simple (recall that simple matrices are square matrices with simple eigenvalues, i.e., all eigenvalues are distinct).  $\bar{M}_{2n}(\mathbb{R})$ and $\bar{M}_{2n,s}(\mathbb{R})$ are generic (open and dense) in $M_{2n}(\mathbb{R})$, and in fact they are (non-connected) manifolds of dimension $2n \times 2n=4n^2$. Similarly, let $A_{2n}(\mathbb{R})$ denote the set of real skew-symmetric $2n \times 2n$ matrices, and $\bar{A}_{2n}(\mathbb{R})$ the subset of full rank such matrices. $\bar{A}_{2n}(\mathbb{R})$ is generic (open and dense) in $A_{2n}(\mathbb{R})$, and in fact it is a (non-connected) manifold of dimension $(2n \times 2n -2n)/2=2n^2 -n$. For a matrix $X$ in $A_{2n}(\mathbb{R})$, we define the principal submatrices $X^{(i)}$  as the upper left corner $2i \times 2i$ submatrices of $X$ (i.e., the sub-matrices formed by rows and columns $1$ to $2i$) from $X^{(1)}$ up to $X^{(n)}=X$. Let $\tilde{A}_{2n}(\mathbb{R})$ be the subset of $\bar{A}_{2n}(\mathbb{R})$ containing matrices $X$ with the property that all $X^{(i)}, i=1,\ldots,n-1$, are full rank. Finally, let $\tilde{M}_{2n,s}(\mathbb{R})$ denote the subset of matrices in $\bar{M}_{2n,s}(\mathbb{R})$ which are admissible. This means that they have the following properties: (i) there is a real invertible $2n \times 2n$ matrix $V$ that puts them in a real Jordan canonical form $J_A$ ($A=V J_A V^{-1}$), which is block-diagonal, with the $1 \times 1$ real blocks before the $2 \times 2$ complex blocks (recall that $A$ is simple, so it has no real Jordan blocks of dimension higher than two), and (ii) $V^{\top} \mathbb{J}_n V \in \tilde{A}_{2n}(\mathbb{R})$. The proof uses arguments inspired by the proof of genericity of simple matrices in the set of all real square matrices of a given dimension, from \cite[Section 5.6]{HSD04}. Also, it relies heavily on methods and results from differential topology. A standard reference for these methods and results, along with terminology and notation, is the book \cite{GP74}. Finally, a crucial argument uses Theorem 5.16 of \cite[Chapter II]{Kato95} and its proof. The proof uses Lemma \ref{Eigenstructure deformation} and Proposition \ref{Generic Property}, and Lemma \ref{J-Grammian function} is needed in the proof of Proposition \ref{Generic Property}. All these results will be proved later on in this appendix.
\begin{lemma}\label{Eigenstructure deformation}
Let $T \in \mathbb{R}^{N \times N}$ be a simple matrix with nonzero eigenvalues. There is a neighborhood of $T$ in $\mathbb{R}^{N \times N}$ such that, every matrix in this neighborhood has eigenvectors and eigenspaces (the latter represented by  projection operators onto the respective eigenspaces) which are continuous functions of their entries, and moreover, their eigenvalues are of the same type as those of $T$.
\end{lemma}
\begin{lemma} \label{J-Grammian function}
Let $F:\bar{M}_{2n}(\mathbb{R})\rightarrow \bar{A}_{2n}(\mathbb{R})$ be defined by $F(V)=V^{\top} \mathbb{J}_n V$. Then, $F$ is onto, and a submersion (see  \cite[Section 1.4]{GP74} for terminology).
\end{lemma}
\begin{proposition} \label{Generic Property}
The set $\mathfrak{V}_{2n}$ of matrices $V \in \bar{M}_{2n}(\mathbb{R})$ such that $V^{\top} \mathbb{J}_n V \in \tilde{A}_{2n}(\mathbb{R})$, is open and dense.
\end{proposition}

We have to show that $\tilde{M}_{2n,s}(\mathbb{R})$ is an open and dense (generic) subset of $M_{2n}(\mathbb{R})$. First, we show that $\tilde{M}_{2n,s}(\mathbb{R})$ is an open set. Consider a matrix $A \in \tilde{M}_{2n,s}(\mathbb{R}) \subset \bar{M}_{2n,s}(\mathbb{R})$, and let $A=V J_A V^{-1}$. The block-diagonal real Jordan canonical form $J_A$ of $A$, has the $1 \times 1$ real blocks before the $2 \times 2$ complex blocks, and no real Jordan blocks of dimension higher than two. Also, $V \in \mathfrak{V}_{2n}$. Applying Lemma \ref{Eigenstructure deformation} to $A$, we conclude that there is a neighborhood $N^{\prime}(A)$ such that for every $\tilde{A} \in N^{\prime}(A)$, $\tilde{A}=\tilde{V} J_{\tilde{A}} \tilde{V}^{-1}$, with $\tilde{V}$ close to $V$, and $J_{\tilde{A}}$ close to $J_A$, and with the same block structure. Let $E=\tilde{V}-V$. Then, $\tilde{V}^{\top} \mathbb{J}_n \tilde{V}=(V+E)^{\top} \mathbb{J}_n (V+E)=V^{\top} \mathbb{J}_n V + E^{\top} \mathbb{J}_n V + V^{\top} \mathbb{J}_n E + E^{\top} \mathbb{J}_n E$. Hence, $\det \big( \tilde{V}^{\top} \mathbb{J}_n \tilde{V} \big)^{(i)}= P^{(i)} (E)$, a multivariate polynomial of degree $4i$ in the entries $E_{jk}$ of $E=\tilde{V}-V$, whose constant term is $\det \big( V^{\top} \mathbb{J}_n V \big)^{(i)}$. However, for matrices $V \in \mathfrak{V}_{2n}$, $\det \big( V^{\top} \mathbb{J}_n V \big)^{(i)} \neq 0, i=1,\ldots,n-1$. By continuity, there exists $\varepsilon_0 >0$, such that $\det \big( \tilde{V}^{\top} \mathbb{J}_n \tilde{V} \big)^{(i)}\neq 0$, for any $E \in M_{2n}(\mathbb{R})$ with $\max_{jk}|E_{jk}|<\varepsilon_0$. By shrinking the neighborhood of $A$ if necessary, we can satisfy $\max_{jk}|\tilde{V}_{jk}-V_{jk}|<\varepsilon_0$, and hence $\tilde{V} \in \mathfrak{V}_{2n}$ and $\tilde{A} \in \tilde{M}_{2n}(\mathbb{R})$. This proves that $\tilde{M}_{2n,s}(\mathbb{R})$ is an open subset of $\bar{M}_{2n,s}(\mathbb{R})$.

To prove that $\tilde{M}_{2n,s}(\mathbb{R})$ is a dense subset of $M_{2n}(\mathbb{R})$, we must prove that every $A \in M_{2n}(\mathbb{R})$ has a $\tilde{A} \in \tilde{M}_{2n,s}(\mathbb{R})$ arbitrarily close to it. Since $\bar{M}_{2n,s}(\mathbb{R})$ is a dense subset of $M_{2n}(\mathbb{R})$, there exists a $\bar{A} \in \bar{M}_{2n,s}(\mathbb{R})$ arbitrarily close to $A$. Let $\bar{A}=\bar{V} J_{\bar{A}} \bar{V}^{-1}$. Then, $J_{\bar{A}}$ is a block-diagonal real Jordan canonical form with no Jordan blocks of dimension higher than two, and can be structured so that it has the $1 \times 1$ real blocks before the $2 \times 2$ complex blocks. Also, $\bar{V} \in \bar{M}_{2n}(\mathbb{R})$. If $\bar{V}$ is not in $\mathfrak{V}_{2n}$, we know from Proposition \ref{Generic Property} that we can find a $\tilde{V}$ arbitrarily close to $\bar{V}$ that is in $\mathfrak{V}_{2n}$. Then, $\tilde{A}=\tilde{V} J_{\bar{A}} \tilde{V}^{-1} \in \tilde{M}_{2n,s}(\mathbb{R})$, and is arbitrarily close to $A$. Hence, $\tilde{M}_{2n,s}(\mathbb{R})$ is a dense subset of $\bar{M}_{2n,s}(\mathbb{R})$, and the theorem is proven. \hfill $\Box$%

\textbf{Proof of Lemma \ref{Eigenstructure deformation}:} Theorem 5.16 of \cite[Chapter II]{Kato95} states that, for a simple matrix in $\mathbb{C}^{N \times N}$, there is a neighborhood of matrices in $\mathbb{C}^{N \times N}$ that contains it, such that the eigenvalues of every matrix in this neighborhood are holomorphic functions of the matrix entries. Furthermore, in the proof of this theorem, it is shown that the eigenspaces of matrices in this neighborhood are also holomorphic functions of their entries. Specializing these results to real matrices, we have that for a simple matrix in $\mathbb{R}^{N \times N}$, there is a neighborhood of matrices in $\mathbb{R}^{N \times N}$ that contains it, such that the eigenvalues and eigenspaces (with eigenspaces being represented by projection operators onto the respective eigenspaces) of every matrix in this neighborhood are analytic functions of its entries. Let $T \in \mathbb{R}^{N \times N}$ be a simple matrix with nonzero eigenvalues, and $T=U J U^{-1}$ a decomposition of it in a Jordan form ($J$ is diagonal with distinct entries). Let also $\tilde{T} \in \mathbb{R}^{N \times N}$ be a matrix in the neighborhood $N(T)$ of $T$ with the aforementioned properties. Taking the entries of $\tilde{T}$ arbitrarily close to those of $T$, the eigenspaces of the two matrices can be made arbitrarily close, as well. Hence, we may change the chosen eigenvectors of $T$ (columns of $U$) to form eigenvectors of $\tilde{T}$ in a continuous way. Then, we may write $\tilde{T}=\tilde{U} \tilde{J} \tilde{U}^{-1}$, where $\tilde{U}$ is arbitrarily close to $U$. Similarly, the diagonal matrix $\tilde{J}$ of eigenvalues of $\tilde{T}$ will be arbitrarily close to $J$. Due to the property of real matrices to have complex eigenvalues in conjugate pairs, and the fact that that $T$ has no zero eigenvalues, there is a neighborhood $N^{\prime}(T) \subseteq N(T)$ such that every $\tilde{T} \in N^{\prime}(T)$ has eigenvalues not only close, but of the same type as $T$. The reason is that, for a pair of complex conjugate eigenvalues to be created (destroyed), two distinct real nonzero eigenvalues must coalesce to a double eigenvalue (be produced by the separation of two equal real eigenvalues). This, however, can be prevented by shrinking the neighborhood of $T$ as much as necessary. \hfill $\Box$%

\textbf{Proof of Lemma \ref{J-Grammian function}:} First, we show that $F$ is properly defined. Obviously, $F(V)$ is antisymmetric. For a $V \in \bar{M}_{2n}(\mathbb{R})$, $\det V \neq 0$, so $\det F(V)=(\det V)^2 \det \mathbb{J}_n = (\det V)^2 \neq 0$, and hence $F(V) \in \bar{A}_{2n}(\mathbb{R})$. Next, we show that $F$ is onto. Let $X \in \bar{A}_{2n}(\mathbb{R})$. From \cite[Subsection 2.5.14]{HJ85}, we know that there exists a $2n \times 2n$ orthogonal matrix $Q$, and a $2n \times 2n$ block-diagonal matrix $\Lambda$ of the form
\begin{eqnarray*}
\Lambda=\diag\left( \left[\begin{array}{cc} 0 & \lambda_1 \\ -\lambda_1 & 0 \\ \end{array}\right], \ldots,
\left[\begin{array}{rr} 0 & \lambda_n \\ -\lambda_n & 0 \\ \end{array}\right]\right),
\end{eqnarray*}
with $\lambda_i \geq 0, i=1,\ldots,n$, such that $X=Q^{\top} \Lambda Q$. Furthermore, the eigenvalues of $X$ are $\pm\imath\lambda_1,\ldots$, $\pm\imath\lambda_n$, and hence $\lambda_i > 0$, for a full rank $X$. It is easy to see that $\Lambda=D^{\top} \mathbb{J}_n D$, with $D=\diag(\sqrt{\lambda_1},\sqrt{\lambda_1},\ldots,\sqrt{\lambda_n},\sqrt{\lambda_n})$. Then, $X=V^{\top} \mathbb{J}_n V$, for $V=DQ$, and $V$ is full rank because $Q$ and $D$ are. Hence, $F$ is onto.

Continuity and differentiability follow from the fact that the entries of $F(V)$ are second order multivariate polynomials in the entries of $V$. Finally, we show that $F$ is a submersion, i.e. that its derivative $DF_{V}: T_{V}\bar{M}_{2n}(\mathbb{R}) \rightarrow T_{F(V)} \bar{A}_{2n}(\mathbb{R})$ is a surjective linear map from the tangent space $T_{V} \bar{M}_{2n}(\mathbb{R})$ of $\bar{M}_{2n}(\mathbb{R})$ at $V$, to the tangent space $T_{F(V)} \bar{A}_{2n}(\mathbb{R})$ of $\bar{A}_{2n}(\mathbb{R})$ at $F(V)$, see \cite[Chapter 1]{GP74} for terminology and notation. Starting from $F(V)=V^{\top} \mathbb{J}_n V$ and ``taking differentials'' of both sides, we have that $dF=(dV)^{\top} \mathbb{J}_n V + V^{\top} \mathbb{J}_n dV$. Hence, for a tangent vector $v \in T_{V}\bar{M}_{2n}(\mathbb{R})$ (infinitesimal variation $dV$ at $V$), we have $DF_{V}(v)=v^{\top} \mathbb{J}_n V + V^{\top} \mathbb{J}_n v$. The tangent space of $\bar{M}_{2n}(\mathbb{R})$ at any point $V$ is isomorphic to $M_{2n}(\mathbb{R})$, and the tangent space of $\bar{A}_{2n}(\mathbb{R})$ at any point $X$ is isomorphic to $A_{2n}(\mathbb{R})$. Hence, $v \in M_{2n}(\mathbb{R})$, and $DF_{V}(v) \in A_{2n}(\mathbb{R})$. Let $w$ be a tangent vector in $T_{F(V)} \bar{A}_{2n}(\mathbb{R})$. To show that $DF_{V}: T_{V}\bar{M}_{2n}(\mathbb{R}) \rightarrow T_{F(V)} \bar{A}_{2n}(\mathbb{R})$ is surjective, we must show that for any such $w$, there exists at least one $v \in T_{V}\bar{M}_{2n}(\mathbb{R})$, such that $DF_{V}(v)=w$. This is equivalent to the equation $v^{\top} \mathbb{J}_n V + V^{\top} \mathbb{J}_n v=w$ having a solution $v \in M_{2n}(\mathbb{R})$ given any $w \in A_{2n}(\mathbb{R})$. Let $v=\mathbb{J}_n (V^{\top})^{-1}\bar{v}$ in that equation (recall that for any $V \in\bar{M}_{2n}(\mathbb{R})$, $V$ is invertible). Then, it reduces to $-\bar{v}+\bar{v}^{\top}=w$, where the antisymmetry of $\mathbb{J}_n$, and the identity $\mathbb{J}_n^2= -I_{2n}$ were used. The general solution of this equation is $\bar{v}=u-\frac{1}{2}w$, where $u$ is any $2n \times 2n$ symmetric matrix. It is to be expected that the solution for $v$ (equivalently for $\bar{v}$) is not unique, since $T_{V} \bar{M}_{2n}(\mathbb{R})$ is a higher dimensional space from $T_{F(V)} \bar{A}_{2n}(\mathbb{R})$. As a matter of fact, the general solution for $v$, $v=\mathbb{J}_n (V^{\top})^{-1}\bar{v}=\mathbb{J}_n (V^{\top})^{-1}\big(u-\frac{1}{2}w\big)$ is parameterized by a $2n \times 2n$ symmetric matrix $u$. The set of such matrices is a linear space of dimension $\frac{1}{2}2n(2n+1)=2n^2+n$, and its dimension is exactly the difference of dimensions of $T_{V}\bar{M}_{2n}(\mathbb{R})$ ($4n^2$) and $T_{F(V)} \bar{A}_{2n}(\mathbb{R})$ ($2n^2 -n$). Hence, we proved that $DF_{V}$ is surjective for every $V \in \bar{M}_{2n}(\mathbb{R})$, i.e. $F$ is a (local) submersion. \hfill $\Box$

\textbf{Proof of Proposition \ref{Generic Property}:} First, we show that the set of $V \in \bar{M}_{2n}(\mathbb{R})$ such that $F(V)=V^{\top} \mathbb{J}_n V \in \tilde{A}_{2n}(\mathbb{R})$, is open in $\bar{M}_{2n}(\mathbb{R})$. Consider such a $V$. Then, $\det F(V)^{(i)}\neq 0, i=1,\ldots,n$. Let $E \in M_{2n}(\mathbb{R})$, and consider $\det F(V+E)^{(i)}$. Since $F(V+E)=(V+E)^{\top} \mathbb{J}_n (V+E)=V^{\top} \mathbb{J}_n V + E^{\top} \mathbb{J}_n V + V^{\top} \mathbb{J}_n E + E^{\top} \mathbb{J}_n E$, we can see that $\det F(V+E)^{(i)}= P^{(i)} (E)$, a multivariate polynomial of degree $4i$ in the entries $E_{jk}$ of $E$, whose constant term is $\det F(V)^{(i)}$. By continuity, there exists $\varepsilon_0 >0$, such that $\det F(V+E)^{(i)}\neq 0$, for any $E \in M_{2n}(\mathbb{R})$ with $|E_{jk}|<\varepsilon_0, j,k=1,2,\ldots,2n$. This proves that, the set of $V \in \bar{M}_{2n}(\mathbb{R})$ such that $F(V)=V^{\top} \mathbb{J}_n V \in \tilde{A}_{2n}(\mathbb{R})$, is open in $\bar{M}_{2n}(\mathbb{R})$.

Now we shall prove that it is dense as well. It suffices to show that, for $V \in \bar{M}_{2n}(\mathbb{R})$ such that $F(V)=V^{\top} \mathbb{J}_n V \in \big( \bar{A}_{2n}(\mathbb{R}) \backslash \tilde{A}_{2n}(\mathbb{R})\big)$, there exists a $\tilde{V} \in \bar{M}_{2n}(\mathbb{R})$ arbitrarily close to $V$ such that $F(\tilde{V})=\tilde{V}^{\top} \mathbb{J}_n \tilde{V} \in \tilde{A}_{2n}(\mathbb{R})$. Since $X=F(V)=V^{\top} \mathbb{J}_n V \in \big(\bar{A}_{2n}(\mathbb{R}) \backslash \tilde{A}_{2n}(\mathbb{R})\big)$, there exists at least one $1 \leq r \leq n-1$, such that $\det X^{(r)}=\det F(V)^{(r)}=0$. Since $X^{(r)}$ is a $2r \times 2r$ real skew-symmetric matrix, there exists an $2r \times 2r$ orthogonal matrix $Q$ such that $X^{(r)} = Q^{\top} S Q$, with
\begin{eqnarray*}
S=\diag\left( \left[\begin{array}{cc} 0 & \nu_1 \\ -\nu_1 & 0 \\ \end{array}\right], \ldots,
\left[\begin{array}{rr} 0 & \nu_r \\ -\nu_r & 0 \\ \end{array}\right]\right),
\end{eqnarray*}
and $\nu_i \geq 0, i=1,\ldots,r$. Then, $\det X^{(r)} = (\det Q)^2 \det S= (\nu_1 \, \ldots \, \nu_r)^2$. Since $\det X^{(r)}=0$, this implies that at least one of the $\nu$'s must be equal to zero. Without loss of generality, we may assume that the first $q$ are equal to zero, $1 \leq q \leq r$. Let $\tilde{S}$ be given by the expression above, where the zero $\nu$'s have been replaced by nonzero $\varepsilon_1,\ldots,\varepsilon_q$:
\begin{eqnarray*}
\tilde{S}&=&\diag\left( \left[\begin{array}{cc} 0 & \varepsilon_1 \\ -\varepsilon_1 & 0 \\ \end{array}\right], \ldots,
\left[\begin{array}{cc} 0 & \varepsilon_q \\ -\varepsilon_q & 0 \\ \end{array}\right],
 \left[\begin{array}{cc} 0 & \nu_{q+1} \\ -\nu_{q+1} & 0 \\ \end{array}\right], \right.\\
&&\quad \left. \ldots,
\left[\begin{array}{cc} 0 & \nu_{r} \\ -\nu_{r} & 0 \\ \end{array}\right] \right).
\end{eqnarray*}
Let also,
\begin{eqnarray*}
\left(\begin{array}{cc}
S & W \\
-W^{\top} & U \\
\end{array}\right) =
\left(\begin{array}{cc}
Q & \mathbf{0} \\
\mathbf{0} & I_{2(n-r)} \\
\end{array}\right) X
\left(\begin{array}{cc}
Q^{\top} & \mathbf{0} \\
\mathbf{0} & I_{2(n-r)} \\
\end{array}\right),
\end{eqnarray*}
and
\begin{eqnarray*}
\tilde{X} = \left(\begin{array}{cc}
Q^{\top} & \mathbf{0} \\
\mathbf{0} & I_{2(n-r)} \\
\end{array}\right)
\left(\begin{array}{cc}
\tilde{S} & W \\
-W^{\top} & U \\
\end{array}\right)
\left(\begin{array}{cc}
Q & \mathbf{0} \\
\mathbf{0} & I_{2(n-r)} \\
\end{array}\right).
\end{eqnarray*}
It is obvious that $\tilde{X}$ can be arbitrarily close to $X$, for small enough $\varepsilon_1,\ldots,\varepsilon_q$, and that $\det \tilde{X}^{(r)}=\det \tilde{S}=(\varepsilon_1 \, \ldots \, \varepsilon_q \, \nu_{q+1} \, \ldots \, \nu_r)^2 \neq 0$. We can also see that $\det \tilde{X}^{(i)} = \tilde{P}^{(i)}(\varepsilon_1,\ldots,\varepsilon_q)$, for $i=1,\ldots,n$, where $\tilde{P}^{(i)}(\varepsilon_1,\ldots,\varepsilon_q)$ is a multivariate polynomial of degree at most $2i$ in the variables $\varepsilon_1,\ldots,\varepsilon_q$, with constant term equal to $\det X^{(i)}$. Hence, for small enough $\varepsilon_1,\ldots,\varepsilon_q$, all the determinants $\det X^{(i)} \neq 0$ remain so for $\tilde{X}$. So, by slightly changing $X$ to $\tilde{X}$, we increased the number of principal submatrices $\tilde{X}^{(i)}$ of full rank by (at least) one, compared with those of $X$. If, for some principal submatrices of $\tilde{X}^{(i)}$ (such that $\det X^{(i)}=0$), we still have $\det \tilde{X}^{(i)} =0$, we may apply the same procedure sequentially, and end up with a matrix $\tilde{X} \in \tilde{A}_{2n}(\mathbb{R})$ arbitrarily close to $X$. From Proposition \ref{J-Grammian function}, $F$ is globally onto. This guarantees that there exists $\tilde{V} \in \bar{M}_{2n}(\mathbb{R})$ such that $\tilde{V}^{\top} \mathbb{J}_n \tilde{V}=F(\tilde{V})=\tilde{X}$. Moreover, $F$ is a submersion at $V$. The Local Submersion Theorem \cite[Section 1.4]{GP74}, guarantees that a neighborhood of $X=F(V)$ (in which we may assume that $\tilde{X}$ belongs to, because we may construct $\tilde{X}$ to be arbitrarily close to to $X$) is the image under $F$ of a neighborhood of $V$. Then, there exists $\tilde{V} \in \bar{M}_{2n}(\mathbb{R})$ in said neighborhood of $V$, such that $\tilde{V}^{\top} \mathbb{J}_n \tilde{V}=F(\tilde{V})=\tilde{X} \in \tilde{A}_{2n}(\mathbb{R})$. Hence, the set of $V \in \bar{M}_{2n}(\mathbb{R})$ such that $F(V)=V^{\top} \mathbb{J}_n V \in \tilde{A}_{2n}(\mathbb{R})$, is also dense in $\bar{M}_{2n}(\mathbb{R})$, and the proposition is proven. \hfill $\Box$

\bibliographystyle{IEEEtran}
\bibliography{ieeeabrv,rip,mjbib2004,irpnew}

\end{document}